\documentclass[review]{elsarticle}

\usepackage{lineno,hyperref}
\modulolinenumbers[5]

\journal{Journal of \LaTeX\ Templates}

\biboptions{authoryear}

\usepackage{amsmath, amsthm}
\usepackage{blindtext, xcolor}
\usepackage{algorithm}
\usepackage{algorithmicx}
\usepackage{algpseudocode}
\usepackage{appendix}
\usepackage{rotating}
\usepackage{float}
\usepackage{comment}
\usepackage{graphicx}
\usepackage{caption,subcaption}
\usepackage{cleveref}
\usepackage{arydshln}
\usepackage{amssymb,mathtools}
\usepackage{lipsum}

\newcommand{\be} {\begin{eqnarray*}}
	\newcommand{\ee} {\end{eqnarray*}}

\theoremstyle{definition}

\newtheorem{theorem}{Theorem}[section]

\newtheorem{proposition}[theorem]{Proposition}

\newcommand{\pkg}[1]{{\fontseries{b}\selectfont #1}}   

\begin{document}

\begin{frontmatter}

\title{Adaptive Bayesian Variable Clustering via Structural Learning of Breast Cancer Data}
\tnotetext[mytitlenote]{Fully documented templates are available in the elsarticle package on \href{http://www.ctan.org/tex-archive/macros/latex/contrib/elsarticle}{CTAN}.}

\author{Elsevier\fnref{myfootnote}}
\address{Radarweg 29, Amsterdam}
\fntext[myfootnote]{Since 1880.}

\author[mymainaddress]{Riddhi Pratim Ghosh}
\ead{RiddhiPratim.Ghosh@Pennmedicine.upenn.edu}

\author[mysecondaryaddress]{Arnab Kumar Maity\corref{mycorrespondingauthor}}
\cortext[mycorrespondingauthor]{Corresponding author}
\ead{Arnab.Maity@pfizer.com}

\author[thirdaddress]{Mohsen Pourahmadi}
\ead{pourahm@stat.tamu.edu}

\author[thirdaddress]{Bani K. Mallick}
\ead{bmallick@stat.tamu.edu}

\address[mymainaddress]{Department of Biostatistics, University of Pennsylvania}
\address[mysecondaryaddress]{Early Clinical Development Oncology Statistics, Pfizer, 10777 Science Center Drive, San Diego, California 92121}
\address[thirdaddress]{Department of Statistics, Texas A\&M University, 3143 TAMU, College Station, Texas 77843}

\begin{abstract}
Clustering of proteins is of interest in cancer cell biology.  This article proposes a hierarchical Bayesian model for protein (variable) clustering hinging on correlation structure.  Starting from a multivariate normal likelihood, we enforce the clustering through prior modeling using angle based unconstrained reparameterization of correlations and assume a truncated Poisson distribution (to penalize the large number of clusters) as prior on the number of clusters. The posterior distributions of the parameters are not in explicit form and we use a reversible jump Markov chain Monte Carlo (RJMCMC) based technique is used to simulate the parameters from the posteriors. The end products of the proposed method are estimated cluster configuration of the proteins (variables) along with the number of clusters. The Bayesian method is flexible enough to cluster the proteins as well as the estimate the number of clusters. The performance of the proposed method has been substantiated with extensive simulation studies and one protein expression data with a hereditary disposition in breast cancer where the proteins are coming from different pathways.
\end{abstract}


\end{frontmatter}


\section{Introduction}

In cell biology different pathways emerge as they play different and critical role in cell functions. Even though the functionality of a cell is an outcome of all the pathway protein expressions as a whole, the individual analysis of each protein has the potential to unveil the complex characterization of the cell biology which is the key to understand the proper cell function \citep{ben1999clustering} .  

The goal of the clustering is to distill the data down to a more comprehensible level subdividing the omics data \citep{d2005does}. In this article we focus on the proteomics data and the interest is to infer about the pathways based on the proteins data via the clustering technique. Clustering of proteins is a form of unsupervised learning where the proteins are grouped on the basis of some similarity measures inherent among them. Such clusters can be mapped to find the appropriate pathway based on the available protein expression data. Given the functions of the proteins which are measured via the RPPA technology based protein expressions, it is of interest to track back the pathways in which the group of proteins belong to, assuming that the pathways do not have an overlap.

A proper clustering method which explicates the pattern involved in the gene expression depending on the over-expression or under-expression of those uncovers the tumor subtypes. For example, \cite{pollack2002microarray} showed that a profiling of DNA copy number variation has the potential to detect more aggressive breast tumors. \cite{washburn2003protein} considered correlation of mRNA and protein expression of amino acid and nucleotide biosynthetic pathway components for clustering. \cite{ben1999clustering} provided an algorithm \textit{Cluster Affinity Search Technique} which uses an affinity measure between nodes of a graph where the genes are represented as the nodes of a graph. In the absence of genuine variable clustering methods, very often traditional data clustering algorithms have been applied to this  setup using brute force \citep{vigneau2003clustering, duda2001pattern} or  ad-hoc algorithms based on aspects of correlation matrices have been proposed. We refer the readers to \cite{jiang2004cluster} for a detailed discussion of various correlation based clustering approaches which have been previously used in literature for analyzing different gene expression data. However, in the current era of next generation sequencing the amount of data that one receives and underlying complexity of the pattern often pose challenges for interpretation and understanding the results, necessitating a proper and meaningful clustering tool.


In this article, our aim is to cluster the proteins, essentially a variable clustering technique which is drastically different from  approaches for clustering observations or subjects. To understand it better, let ${\bf Y}$ denote a $n\times k$ data matrix consisting of $k$ proteins and $n$ patients, represented in the matrix form
\begin{equation}
	\label{data_matrix_eqn}
	\bf Y=
	\bordermatrix{&\text{Protein 1} & \text{Protein 2} & \text{Protein 3} & \text{Protein 4} & \dots &\text{Protein k}\cr
		& y_{11}& y_{12} & y_{13} & y_{14} & \dots  & y_{1k}\cr
		& y_{21} & y_{22} & y_{23} & y_{24} &  \dots  & y_{2k} \cr
		& y_{31} & y_{32} & y_{33} & y_{34} &  \dots & y_{3k}\cr
		&\vdots & \vdots & \vdots & \ddots & \vdots \cr
		& y_{n1} & y_{n2} & y_{n3} & y_{n4}& \dots & y_{nk}}
\end{equation}

From $\eqref{data_matrix_eqn}$, one notes that each of $n$ rows corresponds to one patient and  each of $k$ columns pertains to one protein. A typical data clustering approach partitions the rows of $\bf Y$, i.e. essentially clustering of the patients. We are interested in partitioning the columns of $\bf Y$ which is essentially clustering of the proteins, and correlations between the proteins serve as our main building block to implement the algorithm. In a typical data clustering algorithm we consider how similar the objects are based on a similarity norm (say Euclidean or some other kind of distance). On the contrary, in a variable clustering problem, we are concerned with the correlation among the variables. Hence, highly correlated proteins are more likely to lie in the same cluster. As an example, consider a cluster analysis of a set of proteins which belong to different signaling pathways assuming the pathways are not overlapping. The genetic behaviors control the proliferation of a cell or death of a cell; and depending on signals the proteins receive and send, the cell structure is classified into signaling pathways. In turn, one can assume that the similarly expressed proteins belong to the same pathway which can be recovered via a variable cluster analysis.

Among the different algorithmic clustering techniques commonly used in practice, hierarchical
clustering (agglomerative and divisive approach) and partition methods (K-means clustering)
hing on a distance metric \citep{bibby1979multivariate,friedman2001elements,rokach2005clustering} without assuming any underlying probability model for the clusters. In addition, model based approach usually assumes a mixture model for the data. Even though there is a vast amount of works in the field of data clustering, but the variable clustering problem is at its infancy and has gotten limited attention \citep{bunea2020model}. 
The literature on Bayesian methods for variable clustering is also sparse with a few notable exceptions (\citealt{liechty2004bayesian}, \citealt{palla2012nonparametric}). \cite{palla2012nonparametric} developed a nonparametric Bayes algorithm based on Chinese restaurant process. On the other hand, our method is in the spirit of \cite{liechty2004bayesian} where a parametric model based approach has been considered. A key advantage of our approach is that the number of clusters is assumed to be unknown apriori, and is determined using a reversible jump  Markov Chain Monte Carlo algorithm [RJMCMC] \citep{green1995reversible}.

In this article, our contributions can be summarized as, first, to develop the model-based variable clustering method with block common correlation structures. Second, we propose a novel variable clustering algorithm using the  angular representation of the correlations 
\citep{pinheiro1996unconstrained, rapisarda2007parameterizing, tsay2017modelling, ghosh2020bayesian} and the ensuing angles (hyperspherical coordinates). Third, we elicit substantive prior information on these angles which makes clustering of the variables feasible, a data-driven estimate of number of clusters which traditional algorithms fail to provide.  For the posterior inference, since the angle parameters are badly entangled in the posterior distribution, we resort to the Markov chain Monte Carlo algorithm \citep{tierney1994markov}. For the posterior inference, we resort to the standard RJMCMC techniques as in \citep{green1995reversible,robert2004monte,green2009reversible,fan2011reversible}. The rest of the article is organized as follows. In \ref{model and methods} we review angular reparameterization of a correlation matrix and present clustering model through prior specification on the angles. In \cref{posterior computation} we describe our posterior computation through RJMCMC. Section 4 presents simulation results and clustering of a protein expression data. Finally \cref{conclusion} concludes the article.



\section{Review of angular reparametrization ($\Theta$) of $R$} \label{model and methods}

This section describes connections between the hyperspherical coordinates (angles) and a  correlation matrix $R=(r_{ij})$. 

For a general $k\times k$ correlation matrix $R$ with $1$'s in the diagonal, its Cholesky decomposition is given by $R= BB^{\top}$ where the Cholesky factor $B$ is a lower triangular matrix. Since the rows of $B$ are vectors of unit-length, it turns out that they admit the following representation involving trigonometric functions of some angles  \citep{pinheiro1996unconstrained,rapisarda2007parameterizing}:
\begin{equation} \label{cholesky factor}
	B=
	\begin{bmatrix}
		1& 0 & 0 & 0 & \dots  & 0\\
		c_{21} & s_{21} & 0 & 0 &  \dots  & 0 \\
		c_{31} & c_{32}s_{31} & s_{32}s_{31} & 0 & \dots & 0\\
		c_{41} & c_{42}s_{41} & c_{43}s_{42}s_{41} & \prod_{j=1}^{3}s_{4j} & \dots & 0\\
		\vdots & \vdots & \vdots & \ddots & \vdots \\
		c_{k1} & c_{k2}s_{k1} & c_{k3}s_{k2}s_{k1} & c_{k4}\prod_{j=1}^{3}s_{kj} & \dots & \prod_{j=1}^{k-1}s_{kj}\\
	\end{bmatrix}
\end{equation}
with $c_{ij} = \text{cos}(\theta_{ij})$ and $s_{ij} = \text{sin}(\theta_{ij})$, where the angles $\theta_{ij}$'s are measured in radians, $1\leq j<i\leq k$.
Restricting $\theta_{ij}\in[0,\pi)$ makes the diagonal entries of $ B$ non-negative, and  hence $B$ is unique to which we  associate a $(k-1)\times(k-1)$ lower triangular matrix $\Theta$ with $k(k-1)/2$ angles:
\begin{equation}
	\label{theta_matrix}
	\Theta=
	\begin{bmatrix}
		\theta_{21} & 0 & 0 & \dots  & 0 \\
		\theta_{31} & \theta_{32} & 0 & \dots  & 0 \\
		\vdots & \vdots & \vdots & \ddots & \vdots \\
		\theta_{k1} & \theta_{k2} & \theta_{k3} & \dots  & \theta_{k,k-1}
	\end{bmatrix}
\end{equation}
Note that the $(i, j)$-th element of $\Theta$ is denoted by $\theta_{i+1,j}$ so that $\theta_{ij}$ corresponds to the $(i, j)$-th element of $R$, we refer to $\Theta$ as the {\bf angular matrix} associated to $R$. For further details and applications of these angles, see   \cite{creal2011dynamic}, \cite{zhang2015joint}, \cite{tsay2017modelling}, and \cite{ghosh2020bayesian}	.
One can characterize block diagonal correlation matrices in terms of structured $\Theta$ matrix, which is completely determined by some (\textit{pivotal}) angles.\\

\subsection{Correspondence of clustering between $R$ and $\Theta$}

\begin{proposition}\label{Prop_block_diagonal}
	For a block diagonal correlation matrix $R=\text{block diag}(R_1, R_2,\cdots, R_m)$, consisting of $m$ equicorrelated blocks ($r_i$ for block $R_i$), the corresponding angular matrix $\Theta$ is characterized by only $m$ angles $\theta_1,\theta_2,...,\theta_m$, where $r_i=\text{cos } \theta_i$.
\end{proposition}
\begin{proof}
	See \cref{block_diagonal_1}.
\end{proof}

It follows immediately from \cref{Prop_block_diagonal} that in case of block diagonal correlation matrix, clustering on correlations rendering to $m$ different groups is equivalent to clustering of those $m$ angles by the monotonicity of cosine function. However, this will be impose some conditions on the pivotal angles to maintain positive definiteness. Assuming each block has dimension $k_i$ so that $\sum_{i=1}^{m}k_i=k$, the support of $\theta_i$ is $0< \theta_i< \text{arccos}(1/(k_i-1))$ for $i=1,2,\cdots,m$.\\

\begin{proposition}\label{cluster_separation}
	Suppose that $r_1=\text{cos }\theta_1$, $r_2=\text{cos }\theta_2$. Then $\vert \theta_1-\theta_2\vert \geq \delta$ if and only if $\vert r_1-r_2 \vert\geq \vert 1-\text{cos }\delta\vert$
\end{proposition}
\begin{proof}
	See \cref{block_diagonal_2}.
\end{proof}

\subsection{Likelihood function}
In this article, we assume throughout that the data ${\bf y_1},{\bf y_2},\cdots,{\bf y_n}$ follow a zero mean normal distribution with covariance assumed to be the correlation matrix $R$. As noted in \cref{Prop_block_diagonal}, $R$ can be written as a function of $m$ pivotal angles, $\theta_{pivot}=(\theta_1,\theta_2,...,\theta_m)^{\top}$, and hence denoting the transformation from  $\theta_{pivotal}$ to $R$ by $T$, the likelihood is proportional to,
\begin{align}\label{likelihhood_fn}
	& L({\bf y_1},{\bf y_2},...,{\bf y_n}|\theta_{pivot})\propto  \text{ det }(T(\theta_{pivot}))^{-n/2}\text{ exp }\{-\frac{1}{2}ST^{-1}(\theta_{pivot})\}   
\end{align}
where $S=\sum_{i=1}^{n}{\bf y_i}{\bf y_i}^{\top}$.

\subsection{Prior specification on the angles}\label{prior_set}

The number of clusters m can take any value any value in {1,2,...,k}. Therefore, we assume a truncated Poisson distribution on m. Given m, define $k \times m$ matrix Z whose i-th row corresponds to the allocation of $i$-th variable in one of the $m$ clusters, i.e.
$$
{\bf Z}_{iu}=
\begin{cases}
	1 \;\;\;\ \text{if $i$-th variable belongs to $u$-th cluster} \\
	0 \;\;\ \text{otherwise }
\end{cases}
$$
Since we are assuming that a variable belongs to exactly one cluster, therefore, each row of ${\bf Z}$ contains exactly one 1 and rests are 0s. We assume the following hierarchical prior models for model parameters by assuming a truncated Poisson distribution on number of clusters $m$ to penalize large number of clusters, multinomial distribution on each row of the indicator matrix ${\bf Z}$ and a Dirichlet distribution for the multinomial hyper-parameters. The hierarchical prior structure is succinctly described as,
\begin{align}
	&m \sim truncPois(m;1,k)\\
	& {\bf q}=(q_1,q_2,...,q_m)^{\top} \sim Dirichlet(\alpha_1,\alpha_2,...,\alpha_m)\\
	&{\bf Z}_i \sim Multinomial(1;q_1,q_2,...,q_m) \;\;\ \text{ for }i=1,2,...,k,\\  \nonumber
\end{align}
where $\alpha_i$s are any  positive numbers and $truncPois(m;1,k)$ is a truncated Poisson distribution supported on the integers in between 1 and $k$ (number of proteins or variables) for the number of clusters $m$.
Having sampled ${\bf Z}$, the allocations are determined. Let $k_u$ denote the size of $u$-th cluster, 
\begin{align}
	k_u=\vert \{i:z_{iu}=1\}\vert 
\end{align}
for $u=1,2,\cdots,m$.

Then assume the following prior on  $\theta_{piv}=(\theta_1,\theta_2,...,\theta_m)^{\top}$ in order to shrink them to different values.
\begin{align}\label{prior_theta}
	\theta_{piv}\vert {\bf Z},m, \Lambda = \prod_{u=1}^{m}Q\Big(\theta_u;0,\text{arccos}\big(\frac{1}{k_u-1}\big),\lambda_u\Big)    
\end{align}
where $Q(\theta;0,a,\lambda)$ is the density of truncated wrapped Exponential distribution \citep{mardia2009directional} between 0 and $a$ with parameter $\lambda$.
We are clustering the pivotal angles by introducing wrapped exponential distribution distribution with different parameters. Suppose $\Lambda=(\lambda_1,\lambda_2,...,\lambda_m)^{\top}$ and we sample $\lambda_1,\lambda_2,\cdots,\lambda_m$ in the following manner,
\begin{align}\label{prior_lambda}
	& \lambda_1 \sim N^{+}(\lambda;0,1,0,\infty)\\ \nonumber
	& \lambda_2\vert \lambda_1 \sim N^{+}(\lambda;0,1,\lambda_1,\infty)\\ \nonumber
	& \lambda_i\vert \lambda_{i-1} \sim N^{+}(\lambda;0,1,\lambda_{i-1},\infty)\;\;\; \text{ for }i=2,3,\cdots,m\\ \nonumber,
\end{align}
where $N^{+}(;0,1,a,\infty)$ denotes a truncated  normal distribution on $(0,\infty)$ with mean 0 and variance 1  which has the following density,
\begin{align*}
	f(\lambda;\mu=0,\sigma=1,a,\infty) = \frac{\phi(\lambda)}{1 - \Phi(a) } \;\;\;\ ,
\end{align*}
where $\phi$ and $\Phi$ are the density and distribution function of a standard normal distribution respectively.

The salient features of the prior formulation of $\lambda_i$s given in \eqref{prior_lambda} are the followings: (1) The prior mean for the $i$-th pivotal  angles is $\mathbb{E}\theta_i=\text{arctan}(1/\lambda_i)\;\;\;\ \text{for }i=1,2,...,m$. Since these angles vary in $[0,\pi)$, $\lambda_i$s take value on positive real line. (2) Also $\lambda_i$'s satisfy $\lambda_1<\lambda_2<\lambda_2<...<\lambda_m$, which enforces separation of clusters through prior model.


\section{Posterior computation}\label{posterior computation}
With the likelihoood function \eqref{likelihhood_fn} and prior specified in \ref{prior_set}, the posterior distribution is proportional to
\begin{align}
	\label{posterior}
	& p(\Theta,{\bf Z},\Lambda,m \vert{\bf y_1},{\bf y_2},...,{\bf y_n}) \\ \nonumber
	& \propto L({\bf y_1},{\bf y_2},...,{\bf y_n}|\Theta,{\bf Z},m)\times p(m) \times p({\bf q}|m)\times p({\bf Z}|{\bf q},m) \\ \nonumber
	& \times p(\Theta |\Lambda,{\bf Z},m) \times p(\Lambda) \\ \nonumber
\end{align}
Our goal in this section is to estimate number of clusters $m$ and posterior of ${\bf Z}$. The algorithm is, thus, accomplished by performing a reversible jump Markov chain Monte Carlo(RJMCMC) algorithm. 


From proposed priors, one can note that the clusters are induced by the elements of $\Lambda$, thus, in the following RJMCMC algorithm \citep{green1995reversible,robert2004monte,green2009reversible,fan2011reversible}, at each iteration either one element of $\Lambda$, say $\lambda_j$ is randomly split into $(\lambda_{j_1},\lambda_{j_2})$ (Birth step) or two elements of $\Lambda$ are merged into a single element (Death step). The algorithm is summarized as follows.\\
\begin{itemize}
	\item[] \textit{Step 1}. Initialize $\Theta$, $\Lambda$. In the initialization step, one may assume any block diagonal correlation structure to initialize $\Theta$.
	\item[] \textit{Step 2}. A particular iteration, say $q$-th iteration consists of a Birth step and a Death step.
	\begin{itemize}
		\item[]\textit{Birth Step:}
		Split $\lambda^{(q)}_j$ to $(\lambda^{(q)}_{j_1},\lambda^{(q)}_{j_2})^{\top}$ by $\lambda^{(q)}_{j_1}=\lambda^{(q)}_j+\tau$, $\lambda^{(q)}_{j_2}=\lambda^{(q)}_j-\tau$, where $\tau \sim Unif(-\pi/4,\pi/4)$ and dimension of $\lambda^{(q)}$ is increased by 1 with acceptance probability $\alpha=\text{min}\{1,\frac{p(\Theta^{(q)},\lambda^{(q)}_{j_1},\lambda^{(q)}_{j_2},d(j_1,j_2))}{p(\Theta^{(q)},\lambda^{(q)}_j,d(j))}\times \frac{2}{\pi} \times |\frac{\partial(\lambda^{(q)}_{j_1},\lambda^{(q)}_{j_2})}{\partial(\lambda^{(q)}_j,\tau)}|\}$
		
		\item[] \textit{Death step:}
		Two components $\lambda^{(q)}_{j_1}$ and $\lambda^{(q)}_{j_2}$ are merged to a single component $\lambda^{(q)}_j=(\lambda^{(q)}_{j_1}-\tau+\lambda^{(q)}_{j_2}+\tau)/2$ with 
		acceptance probability $\alpha=\text{min}\{1,\frac{p(\Theta^{(q)},\lambda^{q}_j,d(j))}{p(\Theta^{(q)},\lambda^{(q)}_{j_1},\lambda^{(q)}_{j_2},d(j_1,j_2))}\times \frac{\pi}{2} \times |\frac{\partial(\lambda^{(q)}_j,\tau)}{\partial(\lambda^{(q)}_{j_1},\lambda^{(q)}_{j_2})}|\}$
	\end{itemize}
	\item[] \textit{Step 4}. Step 1, 2 and 3 are repeated as many times as required to ensure convergence and the value of $m$ is determined by which stage is visited maximum number of times, maximum aposteriori estimate (MAP) and posterior estimate of ${\bf Z}$ is obtained by averaging over those stages.
\end{itemize}

\section{Simulations and Data Analyses}\label{simulation and data}
In this section, we compare numerical performance of our Bayesian Variable Clustering (BVC) algorithm with a recent method based on COD (Covariance Difference) of \cite{bunea2020model}, Partitioning Around Medoids (PAM) algorithm which minimizes the Manhattan distance of the data points to the medoids \citep{kaufman2009finding} and the classical or standard K-means clustering algorithm. The performance criterion we use is the proportion of true recovery which is defined for a $k$ variable as 
\begin{equation}  \label{equation_performance}
	\frac{\# \text{ of variables in the true clusters}}{k}.
\end{equation}
COD and PAM have been implemented using the \textsf{R} packages \pkg{cord} \citep{cord} and \pkg{class} \citep{class} available via CRAN and K-means algorithm has been implemented on the transposed data matrix using \textit{kmeans()} function in \textsf{R} software \citep{R}. It is instructive to note that quantity in (\ref{equation_performance}) takes value in the interval [0,1]. As the value approaches to 1, the recovery becomes better.

\subsection{Simulation study} \label{bunea_simulation}
We start with an $m\times m$ matrix $C=B^{\top}B$ where the entries of the random $(m-1)\times m$ matrix $B$ take values $-1,0,1$ with probabilities $0.5\times m^{-1/2}$, $1-m^{-1/2}$ and $0.5\times m^{-1/2}$, respectively, with $m$ being the number of clusters. Next, we consider a balanced case with each group (cluster) of size $k/m$. Let $A=(a_{ij})$ be the $k\times m$ membership matrix with $a_{ij}=1$ if the $i$-th variable belongs to $C_j$ and 0 otherwise. Finally, consider the covariance matrix $\Sigma = ACA^{\top}+\Gamma$ where $\Gamma$ is a diagonal matrix whose entries are random permutations of $\{0.5,0.5+1.5/(k-1),...,2\}$ and  the corresponding correlation matrix $R$. With $k=200, m=4$, we simulate $n$ independent observations from a multivariate normal distribution with mean zero vector and covariance matrix $R$, where we vary $ n $ in {100, 300, 600, 900} to compare BVC, COD and K-means algorithms with respect to cluster recovery criterion in (12). The results presented in Figure 1 shows the superior performance of BVC relative to COD, PAM and K-means as the values stay closer to 1.

\begin{figure}[h]
	\begin{center}
		\includegraphics[width=4in, height = 3in]{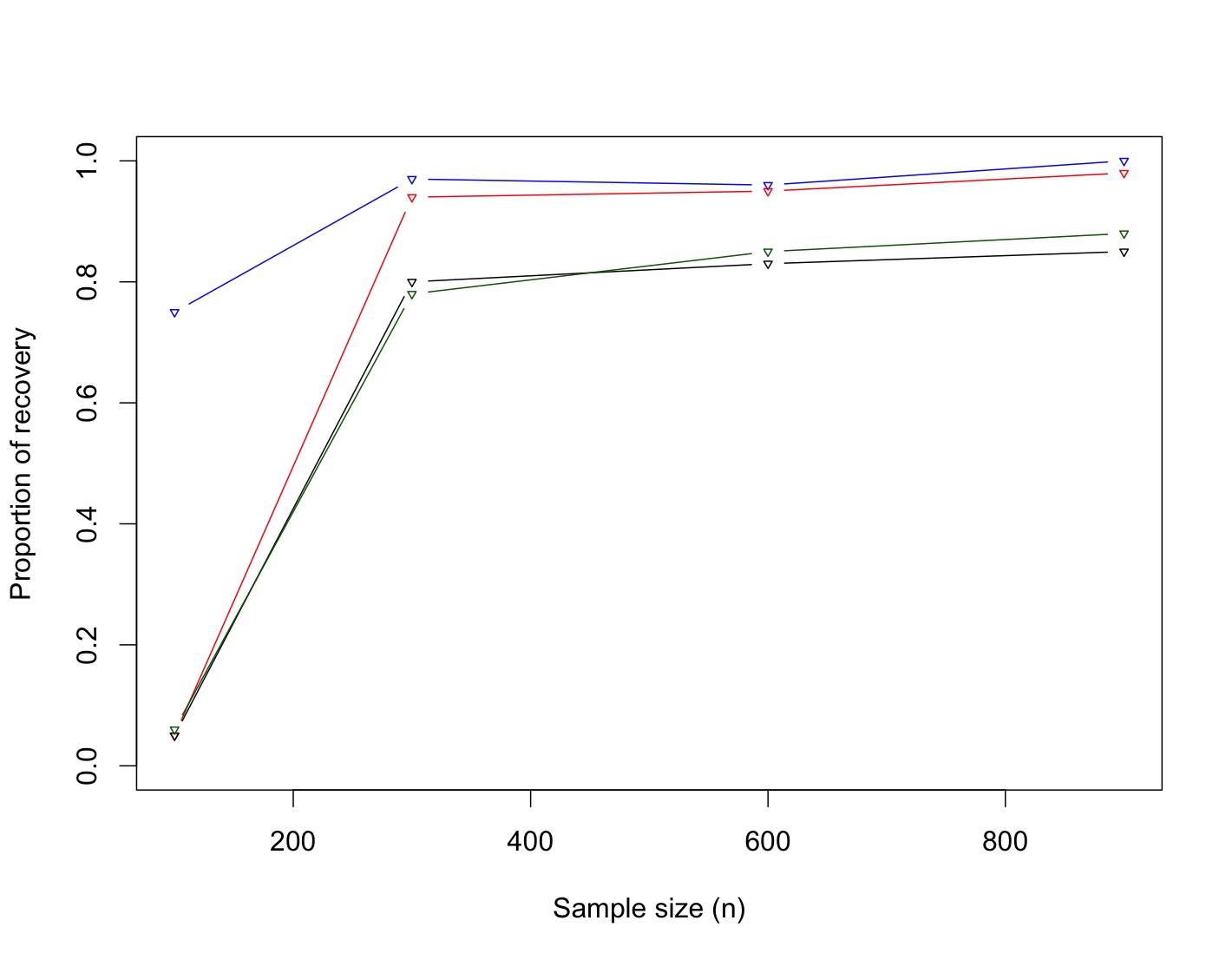}
		\caption{Comparing BVC(blue), COD(red), PAM(green) and K-means(black) for simulation study in \cref{bunea_simulation} }
	\end{center}
\end{figure}

\subsection{Application of Protein clustering to Hereditary Breast Cancer Data}
Breast cancer is one of the most common cancers with a massive number of cases reported. For instance, in 2018, more than 268,000 Americans were estimated to have been diagnosed and 41,000 were estimated to have died from breast cancer related tumors \citep{bray2018global}. The Cancer Genome Atlas: TCGA is the largest available cancer data consortium consisting of parallel mRNA expressions, DNA copy number, methylation expressions, protein expressions, along with clinical variables such as survival or the tumor stages for a total of 33 types of tumors. Among them we consider the information of 222 breast tumor samples; we consider 27 different proteins 4 different pathways (see \cref{pathway_info}). Different subsets of this data has been used in \cite{maity2020bayesian, maity2020bioinformatics} for different purposes.

Applying our BVC algorithm to this data, the MAP estimate of the number of clusters is 4, which is consistent with the number of pathways. 
However, applying the  COD algorithm in \cite{bunea2020model} the estimated number of clusters is 23, much larger than the known value of 4. In \cref{cluster_protein}, we provide the assignments of various proteins in different clusters. Additionally, for the sake of comparison we have also applied the K-means algorithm to this data for $k=4, 23$, respectively, with results reported in the \cref{cluster_protein}. The results suggest that our Bayesian variable clustering (BVC) is performing better to cluster the proteins with respect to pathways. Only misclassified proteins are MAPK\_pT201\_Y204, CD31, CD49b, CDK1. The COD algorithm reports that number of clusters is 23 which appears to be too high since the number of proteins is 27. A possible reason could be this algorithm is meant for high dimensional clustering, it fails to detect clustering configuration in small dimensional cases. Comparisons with standard K-means and PAM algorithm also reveal that these two methods result in more disagreement of the cluster configuration of the proteins according to the pathway information. This apart, K-means and PAM algorithm disagree among themselves, e.g., ER-alpha, JNK\_pT183\_pT185 etc.(Table 2). We have also performed hierarchical clustering on this data with various linkages . The results are presented in Figure 2.

\begin{figure}[!ht]\label{hierarchical_clustering}
	\includegraphics[width= 4.8in]{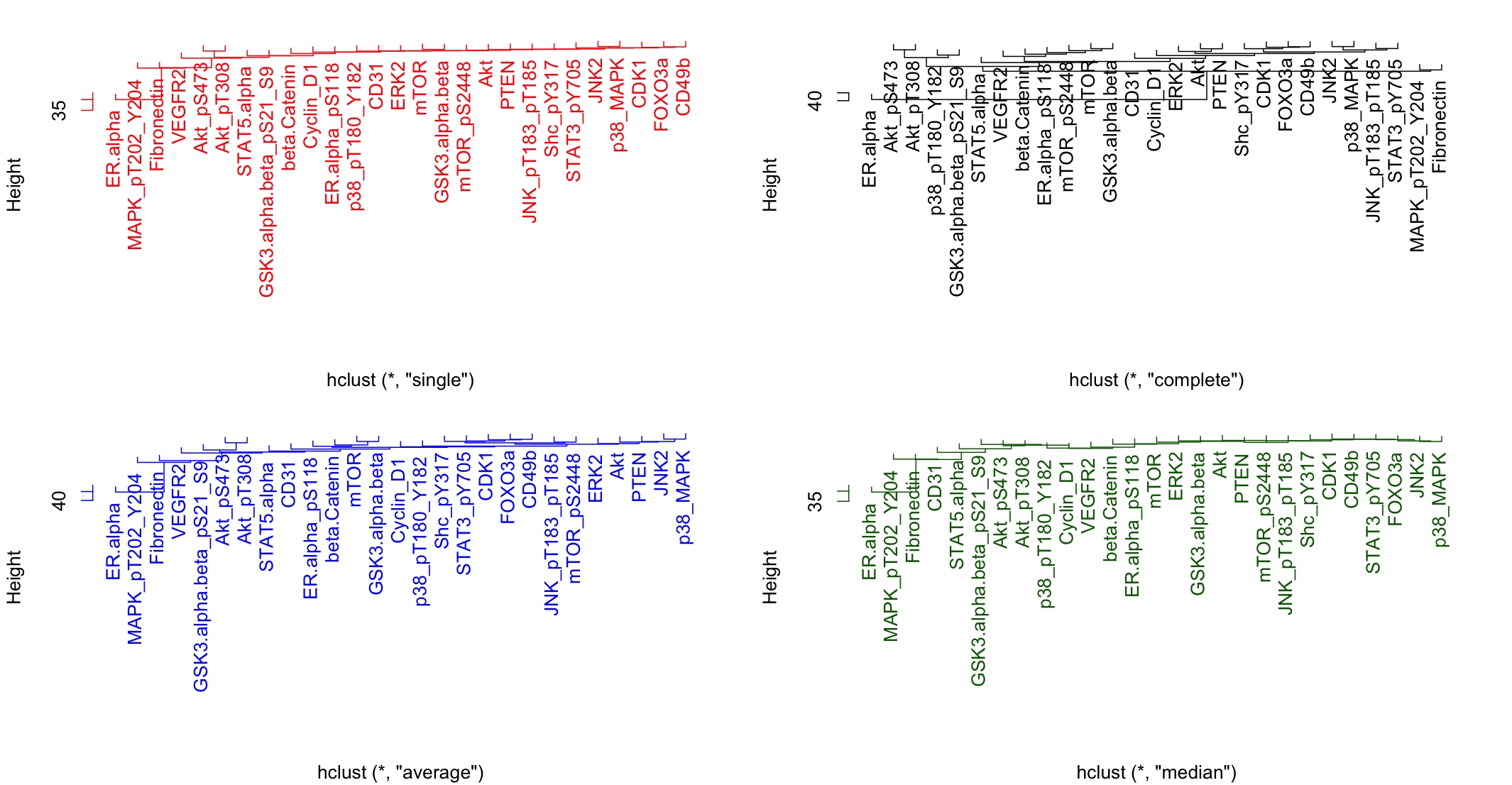}
	\caption{Hierarchical clustering for the protein expression data with four different linkages.}
\end{figure}

\section{Discussion}\label{conclusion}

We have proposed a correlation matrix based Bayesian clustering technique to recover the protein signaling pathways. This method uses angular reparameterization of correlation matrix with the specification of wrapped exponential prior on the angle parameters. Nonetheless, as an alternative, one can use any truncated circular distribution as prior for pivotal angles, for example von-Mises distribution. However, this particular choice produces a mean which has no closed form and as a result our proposed method can not be carried out for a posterior analysis.

A large amount of recent interest is being channelized to analyze the proteomics data directly because direct analysis of proteins has potential to uncover the cell functional characteristics. When it is of interest to find the group of proteins having similar functions which may be evident via their expression measurements then our proposed method can be used to bridge that gap. As mentioned earlier, our method is particularly useful when the number of clusters is not known and hence is learned via the posterior MCMC, which is often the case for the real data where the determining the number of clusters is itself a tedious job.

\section*{Acknowledgment}

The research reported in this paper was supported by grant from the National Institutes of Health (R01CA194391).

\appendix

\section{Appendix}\label{appendix} 

\subsection{Proof of \cref{Prop_block_diagonal}}\label{block_diagonal_1}
\begin{proof}
	The proof uses the Cholesky decomposition of $R$ and the fact that Cholesky factor of a block diagonal matrix is also block diagonal and vice versa.
	
	Hence, the lower triangular Cholesky factor $R$ has the form 
	
	$B=\text{block diag}(B_1,B_2,...,B_m)$, where $B_i$ is an upper triangular Cholesky factor of $R_i$.
	
	The proof will be complete if we can show that Cholesky factor $B$ of a compound symmetric correlation matrix $R$ can be written in terms of only one angle. 
	From the relationship between angles and Cholesky factor, it follows that $\text{cos}(\theta_{i1})=b_{i1}=r$ for $j=2,3,...,k$. Thus $\theta_{i1}$'s are all equal to a common $\theta$. 
	For any $i>j$, the proof follows by induction. For $i>j$, $r_{ij}=\sum_{l=1}^{j-1}b_{il}b_{jl}+b_{jj}b_{ij}$. By induction hypothesis, all the preceding angles and Cholesky factors are functions of $r$. Thus, the first term is a function of $r$. For the second term, we note $b_{jj}=\prod_{l=1}^{j-1}\text{sin}(\theta_{jl})$, which involves all the preceding angles and thus a function of $r$ and $b_{ij}=\text{cos}(\theta_{ij})\prod_{l=1}^{j-1}\text{sin}(\theta_{jl})$. Thus it follows that $\text{cos}(\theta_{ij})$ is a function of $r$ and $\theta_{ij}$ is a function of $\theta$.
	
\end{proof}

\subsection{Proof of \cref{cluster_separation}}\label{block_diagonal_2}
\begin{proof}
	First consider $|\theta_1-\theta_2|=\delta$.
	Note that $|r_1-r_2|=\vert \int_{\theta_1}^{\theta_2}\text{sin }x dx\vert$. Also it is clear that $|r_1-r_2|$ is an increasing function of $|\theta_1-\theta_2|$, since sin is positive in $[0,\pi)$. Now since sin is increasing in $[0,\pi/2]$ and decreasing in $(\pi/2,\pi)$, $|r_1-r_2|$ will take minimum value for $|\theta_1-\theta_2|=\delta$ when $\theta_1=0, \theta_2=\delta$. Thus the minimum value of $|r_1-r_2|$ is $|1-\text{cos}\delta|$.
\end{proof}

\subsection{Pathway Information}\label{pathway_info}
See Table 1 for pathway information of the proteins.
\begin{table}[h]
	\caption{Pathway Protein List}
	\centering
	\begin{tabular}{l|l|l|l}
		\hline
		\hline
		MAP kinase & PI3K/AKT/mTOR & JAK-STAT & Wnt  \\
		\hline
		\hline
		ER-alpha          & AKT          & SHC\_pY317   & CD31 \\
		ER-alpha\_pS118   & AKT\_pS473   & STAT3\_pY705 & CD49b  \\
		ERK2              & AKT\_pT308   & STAT5-alpha  & CDK1 \\
		JNK2              & FOXO3a       &              & Cyclin\_D1 \\
		JNK\_pT183\_pT185 & PTEN         &              & Fibronectin  \\      
		MAPK\_pT202\_Y204 & mTOR         &              & GSK3-alpha-beta \\
		p38\_MAPK & mTOR\_pS2448 &              & GSK3-alpha-beta\_pS21\_S9 \\
		p38\_pT180\_Y182         &              &              & VEGFR2 \\
		&              &              & beta-Catenin  \\
		\hline
		\hline
	\end{tabular}
\end{table}

\subsection{Cluster Assignments of proteins}\label{cluster_protein}  
Table 2 presents the cluster assignments of proteins by BVC, COD, PAM and K-means algorithms.
\begin{table}[h]
	\caption{Cluster comparisons by BVC, COD and K-means}
	\begin{center}
		\begin{tabular}{ c|c|c|c|c|c} 
			\hline \hline
			Protein  &  BVC &  COD & K-means ($m$=4)& PAM ($m$=4) & K-means (k=23)\\
			\hline \hline
			ER-alpha  &  $C_1$ &  $C_{1}$ & $C_3$ & $C_1$ &$C_{2}$\\
			ER-alpha\_pS118  & $C_1$ &  $C_{1}$ & $C_1$ & $C_2$ &$C_{7}$\\
			ERK2  &  $C_1$ &  $C_{2}$ & $C_1$ & $C_2$ &$C_{15}$\\
			JNK2  & $C_1$ &  $C_{5}$ & $C_1$ & $C_2$ &$C_{10}$\\
			JNK\_pT183\_pT185  &  $C_1$ &  $C_{6}$ & $C_4$ & $C_2$& $C_{8}$\\
			MAPK\_pT202\_Y204  &  $C_3$ &  $C_{7}$ & $C_2$& $C_3$ &$C_{112}$\\
			p38\_MAPK  &  $C_1$ & $C_{8}$ & $C_4$ & $C_2$&$C_{9}$\\
			p38\_pT180\_Y182  &  $C_1$ &  $C_{9}$ & $C_4$ &$C_2$&$C_{22}$\\
			\hline \hline
			AKT  &  $C_2$ &  $C_{10}$ & $C_1$ & $C_2$&$C_{11}$\\
			AKT\_pS473  &  $C_2$ &  $C_{11}$ & $C_2$ & $C_4$&$C_{1}$\\
			AKT\_pT308  &  $C_2$ &  $C_{12}$ & $C_2$ &$C_4$ &$C_{1}$\\
			FOXO3a  &  $C_2$ &  $C_{13}$ & $C_4$ & $C_2$&$C_{5}$\\
			PTEN  & $C_2$ &  $C_{14}$ & $C_1$ & $C_2$&$C_{13}$\\
			mTOR  &  $C_2$ &  $C_{15}$ & $C_1$ & $C_2$ &$C_{21}$\\
			mTOR\_pS2448  &  $C_2$ &  $C_{16}$ & $C_1$ & $C_2$ &$C_{20}$\\
			\hline \hline
			SHC\_pY317  &  $C_3$ &  $C_{17}$ & $C_4$ & $C_2$&$C_{17}$\\
			STAT3\_pY705  &  $C_3$ &  $C_{18}$ & $C_4$ &$C_2$ & $C_{23}$\\
			STAT5-alpha  &  $C_3$ &  $C_{2}$ & $C_1$ & $C_2$&$C_{18}$\\
			\hline \hline
			CD31  &  $C_2$ &  $C_{3}$ & $C_4$ & $C_2$&$C_{6}$\\
			CD49b  &  $C_2$ &  $C_{19}$ & $C_4$ & $C_2$&$C_{5}$\\
			CDK1  &  $C_2$ &  $C_{3}$ & $C_4$ & $C_2$&$C_{5}$\\
			Cyclin\_D1  &  $C_4$ &  $C_{20}$ & $C_4$ & $C_2$& $C_{16}$\\
			Fibronectin  &  $C_4$ &  $C_{21}$ & $C_4$ & $C_2$ &$C_{4}$\\
			GSK3-alpha-beta &  $C_4$ &  $C_{4}$ & $C_1$ & $C_2$& $C_{21}$\\
			GSK3-alpha-beta\_pS21\_S9  &  $C_4$ &  $C_{22}$ & $C_2$ & $C_4$&$C_{19}$\\
			VEGFR2  &  $C_4$ &  $C_{23}$ & $C_1$ & $C_2$ &$C_{3}$\\
			beta-Catenin & $C_4$& $C_4$ & $C_1$&$C_2$ & $C_{14}$\\
			\hline \hline
		\end{tabular}
	\end{center}
\end{table}


\bibliography{clustering}

\end{document}